\documentclass[a4paper]{amsart}
\usepackage{amscd}
\usepackage{amsmath}
\usepackage{amssymb}
\usepackage{mathrsfs}
\usepackage{amsthm}
\usepackage{bbm} 
\usepackage{stmaryrd}

\usepackage{lipsum}

\usepackage{xcolor}

\usepackage[T1]{fontenc}

\makeatletter
\@namedef{subjclassname@2020}{%
  \textup{2020} Mathematics Subject Classification}
\makeatother

\numberwithin{equation}{section} \swapnumbers

\newtheorem{satz}{Satz}[section]

\newtheorem{theorem}[satz]{Theorem}
\newtheorem{proposition}[satz]{Proposition}

\newtheorem{lemma}[satz]{Lemma}
\newtheorem{assumption}[satz]{Assumption}

\newtheorem{definition}[satz]{Definition}

\newtheorem{remark}[satz]{Remark}

\newtheorem{example}[satz]{Example}

\newcommand{\bbr}{\mathbb{R}}
\newcommand{\bbe}{\mathbb{E}}
\newcommand{\bbn}{\mathbb{N}}
\newcommand{\bbp}{\mathbb{P}}
\newcommand{\bbq}{\mathbb{Q}}

\newcommand{\bbs}{\mathbb{S}}
\newcommand{\bbx}{\mathbb{X}}

\newcommand{\calc}{\mathscr{C}}

\newcommand{\cale}{\mathscr{E}}
\newcommand{\calf}{\mathscr{F}}

\newcommand{\calk}{\mathscr{K}}

\newcommand{\calm}{\mathscr{M}}

\newcommand{\calv}{\mathscr{V}}

\newcommand{\loc}{{\rm loc}}

\newcommand{\supp}{{\rm supp}}

\newcommand{\adm}{{\rm adm}}
\newcommand{\sfi}{{\rm sf}}

\newcommand{\la}{\langle}
\newcommand{\ra}{\rangle}

\newcommand{\bbI}{\mathbbm{1}}

\newcommand{\bdot}{\boldsymbol{\cdot}}

\newcommand{\IL}{[\![}

\begin{document}

\title[Exploiting arbitrage requires short selling]{Exploiting arbitrage requires short selling}
\author{Eckhard Platen \and Stefan Tappe}
\address{University of Technology Sydney, School of Mathematical and Physical Sciences, Finance Discipline Group, PO Box 123, Broadway, NSW 2007, Australia}
\email{eckhard.platen@uts.edu.au}
\address{Albert Ludwig University of Freiburg, Department of Mathematical Stochastics, Ernst-Zermelo-Stra\ss{}e 1, D-79104 Freiburg, Germany}
\email{stefan.tappe@math.uni-freiburg.de}
\date{12 September, 2022}
\thanks{The authors like to thank Michael Schmutz and Andreas Haier from the Swiss Financial Market Supervisory Authority FINMA for fruitful discussions. The authors are also grateful to Claudio Fontana and Constantinos Kardaras for invaluable comments and remarks. Stefan Tappe gratefully acknowledges financial support from the Deutsche Forschungsgemeinschaft (DFG, German Research Foundation) -- project number 444121509.}
\begin{abstract}
We show that in a financial market given by semimartingales an arbitrage opportunity, provided it exists, can only be exploited through short selling. This finding provides a theoretical basis for differences in regulation for financial services providers that are allowed to go short and those without short sales. The privilege to be allowed to short sell gives access to potential arbitrage opportunities, which creates by design a bankruptcy risk.
\end{abstract}

\maketitle\thispagestyle{empty}

\section{Introduction}\label{sec-intro}

It seems to be a common believe that the exploitation of various forms of arbitrage requires short selling. Indeed, a number of well-known examples exist. However, providing a mathematically rigorous proof for such a statement turns out to be not trivial, as will be shown in this paper. 

There has been always a need for design and regulation of the activities of financial market participants. In particular, the granting of the privilege of short selling to some financial services providers raises the question whether there is a theoretical link between the potential exploitation of arbitrage opportunities and short selling. Hedge funds and investment banks are typical examples for such privileged financial services providers. Short selling occurs when an investor or financial services provider borrows a security and sells it on the open market, planning to buy it back later for less money. The risk of loss on a short sale is theoretically unlimited since the price of the underlying asset can climb to any level. Exploiting arbitrage opportunities can be interpreted as a service to the market because it can remove inconsistencies between asset prices. Market participants with large quantities of assets under management, as pension funds and mutual funds, should not engage in short selling because their risk of bankruptcy could destabilize the market and cause major economic damage. There are also a variety of other reasons for such market participants not to engage in short selling, including higher risk aversion of members, taxation and regulatory constraints aimed at protecting investors.

So far, the literature seems not to have clarified precisely the crucial link between short selling and exploiting arbitrage. Although there are papers dealing with financial markets under short sale constraints (see, for example \cite{Jouini-Kallal, Kardaras-Platen, Pulido, Jarrow-Larsson, JPP, Coculescu, Coculescu-D}), to the best of our knowledge, the literature has not emphasized so far in a mathematically precise manner the important fact that only through short selling one can exploit arbitrage opportunities. The goal of this paper is to fill this gap and show that in a financial market consisting of nonnegative semimartingales an arbitrage opportunity, provided it exists, can only be exploited through short selling. 

Before we provide the mathematical formulation of our results (see Theorems \ref{thm-main}--\ref{thm-main-3c}), we outline the stochastic framework that we consider in this paper; more details and the precise definitions can be found in Section \ref{sec-proof}. Let $\bbs = \{ S^i : i \in I \}$ be a financial market with an arbitrary, possibly infinite index set $I$, consisting of at least two primary security accounts. These accounts could be cum-dividend stocks, savings accounts and combinations of derivatives and other securities offering collateral. We assume that the primary security accounts are nonnegative semimartingales which cannot revive from bankruptcy, and that at least one of these semimartingales is strictly positive. Summing up, for what follows we make the following standing assumption.

\begin{assumption}\label{ass-standing}
Let $\bbs = \{ S^i : i \in I \}$ be a financial market consisting of nonnegative semimartingales which cannot revive from bankruptcy such that $S^i, S_-^i > 0$ for some $i \in I$.
\end{assumption}

Let $\bbp_{\sfi}(\bbs)$ be the space of all self-financing portfolios investing in finitely many assets, and let $\bbp_{\sfi}^{\delta \geq 0}(\bbs)$ be the convex cone of all such portfolios without short selling. We fix a deterministic finite time horizon $T \in (0,\infty)$. Then a portfolio $S^{\delta}$ is an \emph{arbitrage portfolio} if $S_0^{\delta} = 0$ and $S_T^{\delta} \in L_+^0 \setminus \{ 0 \}$, where $L_+^0$ is the convex cone of all equivalence classes of nonnegative random variables. This definition essentially says that an arbitrage portfolio generates under limited liability from zero initial capital some strictly positive wealth. If such an arbitrage portfolio does not exist, then the market satisfies \emph{No Arbitrage} (NA), which means that $\calk_0 \cap L_+^0 = \{ 0 \}$, where $\calk_0$ denotes the convex cone of all outcomes of self-financing portfolios starting at zero. Often, one assumes that an arbitrage portfolio is nonnegative (see, for example \cite{Platen}) or admissible (see, for example \cite{DS-94} or \cite{DS-book}). For the sake of generality, in our upcoming first result we do not assume that an arbitrage portfolio has to be nonnegative or admissible.

\begin{theorem}\label{thm-main}
For every arbitrage portfolio $S^{\delta} \in \bbp_{\sfi}(\bbs)$ we have $S^{\delta} \notin \bbp_{\sfi}^{\delta \geq 0}(\bbs)$.
\end{theorem}

The above theorem states the fact that exploiting arbitrage requires short selling. Although this result seems intuitively plausible, the mathematical proof (provided in Section \ref{sec-proof}, also for the upcoming Theorems \ref{thm-main-2}--\ref{thm-main-3c}) is not obvious in the present semimartingale market framework. The proof is straightforward if the arbitrage portfolio goes negative in between (see Lemma \ref{lemma-neg-port}), or if we consider financial market models in discrete time with strictly positive primary security accounts (see Proposition \ref{prop-discrete}). However, in the present continuous time framework, where we also consider nonnegative arbitrage portfolios, deeper results from martingale theory are required for the proof.

Similar results also hold true for weaker forms of arbitrage opportunities than the arbitrage portfolios considered above, such as, so-called, arbitrages of the first kind or various types of free lunches, as defined later on. In the following results (Theorem \ref{thm-main-2}--\ref{thm-main-3b}) we will consider a sequence $(S^{\delta_n})_{n \in \bbn} \subset \bbp_{\sfi}(\bbs)$ of self-financing portfolios, and assume that $\bigcup_{n \in \bbn} \supp(\delta_n)$ is finite. Here, for every $n \in \bbn$ the support of $\delta_n$, denoted by $\supp(\delta_n)$, consists of all indices of those primary security accounts in which the trading strategy $\delta_n$ invests. We emphasize that the aforementioned condition that $\bigcup_{n \in \bbn} \supp(\delta_n)$ is finite is \emph{always} satisfied if the index set $I$ is finite; that is, the market $\bbs$ consists of finitely many assets.

A sequence $(S^{\delta_n})_{n \in \bbn} \subset \bbp_{\sfi}(\bbs)$ of self-financing portfolios is called an \emph{arbitrage of the first kind} if $S_0^{\delta_n} \downarrow 0$ and there exists a random variable $\xi \in L_+^0 \setminus \{ 0 \}$ such that $S_T^{\delta_n} \geq \xi$ for all $n \in \bbn$. Typically, the sequence $(S^{\delta_n})_{n \in \bbn}$ is assumed to be nonnegative; see, for example \cite{KKS}. If a nonnegative arbitrage of the first kind does not exist, then we have $p(\xi) > 0$ for each $\xi \in L_+^0 \setminus \{ 0 \}$, where
\begin{align*}
p(\xi) := \inf \{ S_0^{\delta} : S^{\delta} \in \bbp_{\sfi}(\bbs) \text{ such that } S^{\delta} \geq 0 \text{ and } S_T^{\delta} \geq \xi \}
\end{align*}
denotes the superreplication price of $\xi$. This just means that the market satisfies \emph{No Arbitrage of the First Kind} (NA$_1$). For the sake of generality, in our upcoming result we do not assume that an arbitrage of the first kind has to remain always nonnegative.

\begin{theorem}\label{thm-main-2}
Let $(S^{\delta_n})_{n \in \bbn} \subset \bbp_{\sfi}(\bbs)$ be an arbitrage of the first kind such that $\bigcup_{n \in \bbn} \supp(\delta_n)$ is finite. Then there exists an index $n_0 \in \bbn$ such that $S^{\delta_n} \notin \bbp_{\sfi}^{\delta \geq 0}(\bbs)$ for all $n \geq n_0$.
\end{theorem}

A sequence $(S^{\delta_n})_{n \in \bbn}$ of self-financing portfolios is called an \emph{asymptotic arbitrage of the first kind} if $S_0^{\delta_n} \downarrow 0$ and there exists a sequence of random variables $(\xi_n)_{n \in \bbn} \subset L_+^0$ such that $S_T^{\delta_n} \geq \xi_n$ for all $n \in \bbn$ and $\limsup_{n \to \infty} \bbp(\xi_n \geq 1) > 0$. If we even have $\limsup_{n \to \infty} \bbp(\xi_n \geq 1) = 1$, then one also calls it a \emph{strong asymptotic arbitrage of the first kind}; cf., e.g. \cite{Kabanov-Kramkov-1} or \cite{Kabanov-Kramkov-2}. Typically, the sequence $(S^{\delta_n})_{n \in \bbn}$ is assumed to be nonnegative; see, for example \cite{KKS}. If a nonnegative arbitrage of the first kind does not exist, then for every sequence $(S^{\delta_n})_{n \in \bbn}$ of nonnegative self-financing portfolios and every sequence of random variables $(\xi_n)_{n \in \bbn} \subset L_+^0$ such that $S_T^{\delta_n} \geq \xi_n$ for all $n \in \bbn$ we have $\lim_{n \to \infty} \bbp(\xi_n \geq 1) = 0$, or equivalently $\xi_n \overset{\bbp}{\to} 0$; cf. \cite[Lemma 5.10]{Platen-Tappe-tvs}. This just means that the market satisfies \emph{No Asymptotic Arbitrage of the First Kind} (NAA$_1$). It is well known that NAA$_1$ is equivalent to \emph{No Unbounded Profit with Bounded Risk} (NUPBR); see, for example \cite{KKS}. For the sake of generality, in our upcoming result we do not assume that an asymptotic arbitrage of the first kind has to be nonnegative.

\begin{theorem}\label{thm-main-2b}
Let $(S^{\delta_n})_{n \in \bbn} \subset \bbp_{\sfi}(\bbs)$ be an asymptotic arbitrage of the first kind such that $\bigcup_{n \in \bbn} \supp(\delta_n)$ is finite. Then we have $S^{\delta_n} \notin \bbp_{\sfi}^{\delta \geq 0}(\bbs)$ for infinitely many $n \in \bbn$.
\end{theorem}

Note that for a given asymptotic arbitrage of the first kind the sequence $(\bbp(\xi_n \geq 1))_{n \in \bbn}$ does not need to be convergent. In this paper, we suggest to call it a \emph{limit arbitrage of the first kind} if the above condition $\limsup_{n \to \infty} \bbp(\xi_n \geq 1) > 0$ is replaced by the stronger condition $\lim_{n \to \infty} \bbp(\xi_n \geq 1) > 0$, implicitly meaning that the limit exists. In the case of a limit arbitrage of the first kind, the previous result can be improved as follows.

\begin{theorem}\label{thm-main-2b-v2}
Let $(S^{\delta_n})_{n \in \bbn} \subset \bbp_{\sfi}(\bbs)$ be a limit arbitrage of the first kind such that $\bigcup_{n \in \bbn} \supp(\delta_n)$ is finite. Then there exists an index $n_0 \in \bbn$ such that $S^{\delta_n} \notin \bbp_{\sfi}^{\delta \geq 0}(\bbs)$ for all $n \geq n_0$.
\end{theorem}

We call a sequence $(S^{\delta_n})_{n \in \bbn}$ of self-financing portfolios with $S_0^{\delta_n} = 0$ for all $n \in \bbn$ a \emph{free lunch with vanishing risk} if there exist a random variable $\xi \in L_+^{\infty} \setminus \{ 0 \}$ and a sequence $(\xi_n)_{n \in \bbn} \subset L^{\infty}$ such that $\| \xi_n - \xi \|_{L^{\infty}} \to 0$ and $S_T^{\delta_n} \geq \xi_n$ for each $n \in \bbn$, where $L^{\infty}$ denotes the Banach space of all equivalence classes of bounded random variables, and where $L_+^{\infty}$ is the convex cone of all equivalence classes of nonnegative bounded random variables. Typically, the sequence $(S^{\delta_n})_{n \in \bbn}$ is assumed to consist of admissible portfolios, which can become negative but remain above a critical level; see, for example \cite{DS-94} or \cite{DS-book}. If an admissible free lunch with vanishing risk does not exist, then we have $\overline{\calc} \cap L_+^0 = \{ 0 \}$, where $\calc := (\calk_0 - L_+^0) \cap L^{\infty}$ with $\calk_0$ denoting the convex cone of all outcomes of admissible self-financing portfolios starting at zero, and where $\overline{\calc}$ denotes the closure of $\calc$ with respect to the norm topology on $L^{\infty}$. This just means that the market satisfies the \emph{No Free Lunch with Vanishing Risk} (NFLVR) condition. In our following result we do not assume that a free lunch with vanishing risk has to be admissible.

\begin{theorem}\label{thm-main-3}
Let $(S^{\delta_n})_{n \in \bbn} \subset \bbp_{\sfi}(\bbs)$ be a free lunch with vanishing risk such that $\bigcup_{n \in \bbn} \supp(\delta_n)$ is finite. Then there exists an index $n_0 \in \bbn$ such that $S^{\delta_n} \notin \bbp_{\sfi}^{\delta \geq 0}(\bbs)$ for all $n \geq n_0$.
\end{theorem}

A sequence $(S^{\delta_n})_{n \in \bbn}$ of self-financing portfolios with $S_0^{\delta_n} = 0$ for all $n \in \bbn$ is called a \emph{free lunch with bounded risk} if there exist a random variable $\xi \in L_+^{\infty} \setminus \{ 0 \}$ and a sequence $(\xi_n)_{n \in \bbn} \subset L^{\infty}$ such that $\xi_n \overset{w^*}{\to} \xi$ and $S_T^{\delta_n} \geq \xi_n$ for each $n \in \bbn$, where $\xi_n \overset{w^*}{\to} \xi$ denotes weak-star convergence in $L^{\infty}$; that is
\begin{align*}
\bbe[\xi_n \eta] \to \bbe[\xi \eta] \quad \text{for all $\eta \in L^1$.}
\end{align*}
Typically, the sequence $(S^{\delta_n})_{n \in \bbn}$ is assumed to consist of admissible portfolios; see, for example \cite{DS-94} or \cite{DS-book}. If an admissible free lunch with bounded risk does not exist, then the market satisfies \emph{No Free Lunch with Bounded Risk} (NFLBR), which means that $\widetilde{\calc} \cap L_+^0 = \{ 0 \}$, where $\widetilde{\calc}$ denotes the closure of $\calc$ with respect to the sequential weak-star topology on $L^{\infty}$. In our following result we do not assume that a free lunch with bounded risk has to be admissible.

\begin{theorem}\label{thm-main-3b}
Suppose that $(S_T^i)^{-1} \in L^{\infty}$, where the index $i \in I$ stems from Assumption \ref{ass-standing}. Furthermore, let $(S^{\delta_n})_{n \in \bbn} \subset \bbp_{\sfi}(\bbs)$ be a free lunch with bounded risk such that $\bigcup_{n \in \bbn} \supp(\delta_n)$ is finite. Then there exists an index $n_0 \in \bbn$ such that $S^{\delta_n} \notin \bbp_{\sfi}^{\delta \geq 0}(\bbs)$ for all $n \geq n_0$.
\end{theorem}

This result can be generalized as follows. Given an arbitrary directed set $(J,\leq)$, a net $(S^{\delta_j})_{j \in J}$ of self-financing portfolios with $S_0^{\delta_j} = 0$ for all $j \in J$ is called a \emph{free lunch} if there exist a random variable $\xi \in L_+^{\infty} \setminus \{ 0 \}$ and a net $(\xi_j)_{j \in J} \subset L^{\infty}$ such that $\xi_j \overset{w^*}{\to} \xi$ and $S_T^{\delta_j} \geq \xi_j$ for each $j \in J$, where $\xi_j \overset{w^*}{\to} \xi$ denotes weak-star convergence in $L^{\infty}$; that is
\begin{align*}
\bbe[\xi_j \eta] \to \bbe[\xi \eta] \quad \text{for all $\eta \in L^1$.}
\end{align*}
Typically, the net $(S^{\delta_j})_{j \in J}$ is assumed to consist of admissible portfolios; see, for example \cite{DS-book}. If an admissible free lunch does not exist, then the market satisfies \emph{No Free Lunch} (NFL), which means that $\overline{\calc}^* \cap L_+^0 = \{ 0 \}$, where $\overline{\calc}^*$ denotes the closure of $\calc$ with respect to the weak-star topology on $L^{\infty}$. In our following result we do not assume that a free lunch has to be admissible. Similarly as for the previous results, the following condition that $\bigcup_{j \in J} \supp(\delta_j)$ is finite is \emph{always} satisfied if the index set $I$ is finite; that is, the market $\bbs$ consists of finitely many assets.

\begin{theorem}\label{thm-main-3c}
Suppose that $(S_T^i)^{-1} \in L^{\infty}$, where the index $i \in I$ stems from Assumption \ref{ass-standing}. Furthermore, let $(S^{\delta_j})_{j \in J} \subset \bbp_{\sfi}(\bbs)$ be a free lunch such that $\bigcup_{j \in J} \supp(\delta_j)$ is finite. Then there exists an index $j_0 \in J$ such that $S^{\delta_j} \notin \bbp_{\sfi}^{\delta \geq 0}(\bbs)$ for all $j \geq j_0$.
\end{theorem}

Let us briefly outline the essential steps for the proofs of Theorems \ref{thm-main}--\ref{thm-main-3c}. A crucial concept for the proofs is that of an \emph{equivalent supermartingale deflator} (ESMD). First, we show that the existence of an ESMD implies the statements of Theorems \ref{thm-main}--\ref{thm-main-3c} (see Propositions \ref{prop-ESMD-NA}--\ref{prop-ESMD-NA-4}). In a second step, we show that in the present setting an ESMD always exists; see Proposition \ref{prop-ESMD-exists}. This result seems to be of independent interest. The main ideas for the proof include a change of num\'{e}raire, a result about ESMDs from \cite{Kardaras-Platen} and a $\sigma$-localization technique from \cite{Kallsen}.

We mention that supermartingale deflators and the related concept of a supermartingale measure have been used in various contexts in the mathematical finance literature; in particular in connection with short sale constraints, including \cite{Jouini-Kallal, Kardaras-Platen, Pulido, Jarrow-Larsson} and \cite{Coculescu}. The papers \cite{Kardaras-Platen} and \cite{Pulido} present versions of the fundamental theorem of asset pricing with short sale prohibitions. Further references, where supermartingales and supermartingale deflators are studied with a focus on applications in finance, include \cite{Zitkovic-BS, Kardaras-13} and the recent paper \cite{Harms}.

The remainder of this paper is organized as follows. In Section \ref{sec-proof} we provide the proofs of our results, in Section \ref{sec-discrete} we briefly comment on discrete time models, and in Section \ref{sec-examples} we present several examples and illustrate how our examples with arbitrage portfolios which go negative link to asset and money market bubbles.

\section{Proofs of the results}\label{sec-proof}

In this section we provide the proofs of Theorems \ref{thm-main}--\ref{thm-main-3c}. First, we introduce the precise mathematical framework. Let $T \in (0,\infty)$ be a fixed finite time horizon, and let $(\Omega,\calf,(\calf_t)_{t \in [0,T]},\bbp)$ be a stochastic basis satisfying the usual conditions; see \cite[Def. I.1.3]{Jacod-Shiryaev}. Furthermore, we assume that $\calf_0$ is $\bbp$-trivial. Then every $\calf_0$-measurable random variable is $\bbp$-almost surely constant. As mentioned in Section \ref{sec-intro}, we consider a financial market $\bbs = \{ S^i : i \in I \}$ consisting of nonnegative semimartingales for an index set $I$ with $|I| \geq 2$. Apart from the latter condition of having at least two elements, the index set $I$ is arbitrary; in particular, it may be infinite, and thus possibly even uncountable. We assume that for each $i \in I$ the semimartingale $S^i$ cannot revive from bankruptcy, which means that $S_t^i = 0$ for all $t \geq \tau^i$, where $\tau^i$ denotes the bankruptcy time of $S^i$ given by
\begin{align*}
\tau^i := \inf \{ t \in \bbr_+ : S_{t-}^i = 0 \text{ or } S_t^i = 0 \}.
\end{align*}
Furthermore, we assume that $S^i, S_-^i > 0$ for some $i \in I$.

For an $\bbr^d$-valued semimartingale $X$ we denote by $L(X)$ the set of all $X$-integrable processes in the sense of vector integration; see \cite{Shiryaev-Cherny} or \cite[Sec. III.6]{Jacod-Shiryaev}. For $\delta \in L(X)$ we denote by $\delta \bdot X$ the stochastic integral according to \cite{Shiryaev-Cherny}.

For a finite set $F \subset I$ we define the multi-dimensional semimartingale $S^F := (S^i)_{i \in F}$. We call a process $\delta = (\delta^i)_{i \in I}$ a \emph{strategy} for $\bbs$ if its \emph{support}
\begin{align}\label{support}
F := \supp(\delta) := \{ i \in I : \delta^i \neq 0 \}
\end{align}
is finite and we have $\delta^F \in L(S^F)$, where $\delta^F := (\delta^i)_{i \in F}$. We will also use the notation $\Delta(\bbs)$ for the space of all strategies for $\bbs$. For a strategy $\delta \in \Delta(\bbs)$ we define the \emph{portfolio} $S^{\delta} := \delta \cdot S$, where we use the short-hand notation $\delta \cdot S := \sum_{i \in F} \delta^i S^i$ with $F := \supp(\delta)$. A strategy $\delta \in \Delta(\bbs)$ and the corresponding portfolio $S^{\delta}$ are called \emph{self-financing} for $\bbs$ if $S^{\delta} = S_0^{\delta} + \delta \bdot S$, where $\delta \bdot S := \delta^F \bdot S^F$ is the stochastic integral according to \cite{Shiryaev-Cherny} with $F := \supp(\delta)$. We denote by $\Delta_{\sfi}(\bbs)$ the space of all self-financing strategies.

Let $\bbp_{\sfi}(\bbs)$ be the space of all self-financing portfolios $S^{\delta}$. We denote by $\bbp_{\sfi}^+(\bbs)$ the convex cone of all nonnegative self-financing portfolios $S^{\delta} \geq 0$, and we denote by $\bbp_{\sfi}^{\adm}(\bbs)$ the convex cone of all admissible self-financing portfolios $S^{\delta} \geq -a$ for some $a \in \bbr_+$. Moreover, we denote by $\bbp_{\sfi}^{\delta \geq 0}(\bbs)$ the convex cone of all self-financing portfolios $S^{\delta} \in \bbp_{\sfi}(\bbs)$ such that $\delta \geq 0$, which means that $\delta^i \geq 0$ for all $i \in I$. A self-financing portfolio $S^{\delta}$ is called an \emph{arbitrage portfolio} if $S_0^{\delta} = 0$ and $S_T^{\delta} \in L_+^0 \setminus \{ 0 \}$.

Recall that we assume that $S^i, S_-^i > 0$ for some $i \in I$. The proofs of the following two auxiliary results are straightforward, and therefore omitted.

\begin{lemma}\label{lemma-fin-market-1}
For each $\delta \in \Delta_{\sfi}(\bbs)$ we have $\delta^F \in \Delta_{\sfi}(\bbs^F)$, where $F := \supp(\delta)$, and $\bbs^F$ denotes the finite market
\begin{align*}
\bbs^F := \{ S^j : j \in F \cup \{ i \} \}.
\end{align*}
\end{lemma}

\begin{lemma}\label{lemma-fin-market-2}
Let $(\delta_n)_{n \in \bbn} \subset \Delta_{\sfi}(\bbs)$ be a sequence of self-financing strategies such that $F := \bigcup_{n \in \bbn} \supp(\delta_n)$ is finite. Then we also have $\delta_n^F \in \Delta_{\sfi}(\bbs^F)$ for each $n \in \bbn$, where $\bbs^F$ denotes the finite market
\begin{align*}
\bbs^F := \{ S^j : j \in F \cup \{ i \} \}.
\end{align*}
\end{lemma}

\begin{remark}\label{rem-fin-market}
In view of Lemmas \ref{lemma-fin-market-1} and \ref{lemma-fin-market-2}, from now on we may assume that we have a finite market $\bbs = \{ S^1,\ldots,S^d \}$ consisting of nonnegative semimartingales for some $d \in \bbn$ with $d \geq 2$, such that the semimartingales $S^1,\ldots,S^{d-1} \geq 0$ cannot revive from bankruptcy, and we have $S^d, S_-^d > 0$.
\end{remark}

\begin{lemma}\label{lemma-neg-port}
For every self-financing portfolio $S^{\delta} \in \bbp_{\sfi}(\bbs)$ the property $S^{\delta} \in \bbp_{\sfi}^{\delta \geq 0}(\bbs)$ implies $S^{\delta} \in \bbp_{\sfi}^+(\bbs)$.
\end{lemma}

\begin{proof}
Since $S^i \geq 0$ for all $i=1,\ldots,d$, the property $\delta^i \geq 0$ for all $i=1,\ldots,d$ implies $S^{\delta} = \delta \cdot S = \sum_{i=1}^d \delta^i S^i \geq 0$.
\end{proof}

\begin{definition}
Let $\bbx$ be a family of semimartingales, and let $Z$ be a semimartingale such that $Z,Z_- > 0$. We call $Z$ an \emph{equivalent supermartingale deflator (ESMD)} for $\bbx$ if $X Z$ is a supermartingale for all $X \in \bbx$.
\end{definition}

The following Propositions \ref{prop-ESMD-NA}--\ref{prop-ESMD-NA-4} show the statements of Theorems \ref{thm-main}--\ref{thm-main-3c}, provided that an ESMD $Z$ for $\bbp_{\sfi}^{\delta \geq 0}(\bbs) \cap \bbp_{\sfi}^{+}(\bbs)$ exists.

\begin{proposition}\label{prop-ESMD-NA}
Suppose that an ESMD $Z$ for $\bbp_{\sfi}^{\delta \geq 0}(\bbs) \cap \bbp_{\sfi}^{+}(\bbs)$ exists. Then for every arbitrage portfolio $S^{\delta} \in \bbp_{\sfi}(\bbs)$ we have $S^{\delta} \notin \bbp_{\sfi}^{\delta \geq 0}(\bbs)$.
\end{proposition}

\begin{proof}
Let $S^{\delta} \in \bbp_{\sfi}^{\delta \geq 0}(\bbs) \cap \bbp_{\sfi}^{+}(\bbs)$ be a self-financing portfolio such that $S_0^{\delta} = 0$ and $\xi := S_T^{\delta} \in L_+^0$. Since $Z$ is an ESMD for $\bbp_{\sfi}^{\delta \geq 0}(\bbs) \cap \bbp_{\sfi}^{+}(\bbs)$, the process $S^{\delta} Z$ is a nonnegative supermartingale. By Doob's optional stopping theorem we obtain
\begin{align*}
\bbe[\xi Z_T] = \bbe[S_T^{\delta} Z_T] \leq \bbe[S_0^{\delta} Z_0] = 0,
\end{align*}
and hence, we deduce that $\bbe[\xi Z_T] = 0$. Since $\xi \geq 0$ and $\bbp(Z_T > 0) = 1$, this shows $\xi = 0$. Consequently, the portfolio $S^{\delta}$ cannot be an arbitrage portfolio. Together with Lemma \ref{lemma-neg-port}, this completes the proof.
\end{proof}

\begin{proposition}\label{prop-ESMD-NA-2}
Suppose that an ESMD $Z$ for $\bbp_{\sfi}^{\delta \geq 0}(\bbs) \cap \bbp_{\sfi}^{+}(\bbs)$ exists. Let
\begin{align*}
(S^{\delta_n})_{n \in \bbn} \subset \bbp_{\sfi}(\bbs) 
\end{align*}
be an arbitrage of the first kind. Then there exists an index $n_0 \in \bbn$ such that $S^{\delta_n} \notin \bbp_{\sfi}^{\delta \geq 0}(\bbs)$ for all $n \geq n_0$.
\end{proposition}

\begin{proof}
We have $S_0^{\delta_n} \downarrow 0$ and there exists a random variable $\xi \in L_+^0 \setminus \{ 0 \}$ such that $S_T^{\delta_n} \geq \xi$ for all $n \in \bbn$. Suppose, contrary to the assertion above, there is a subsequence $(S^{\delta_{n_k}})_{k \in \bbn}$ such that $S^{\delta_{n_k}} \in \bbp_{\sfi}^{\delta \geq 0}(\bbs)$ for all $k \in \bbn$. Let $k \in \bbn$ be arbitrary. By Lemma \ref{lemma-neg-port} we have $S^{\delta_{n_k}} \in \bbp_{\sfi}^+(\bbs)$. Since $Z$ is an ESMD for $\bbp_{\sfi}^{\delta \geq 0}(\bbs) \cap \bbp_{\sfi}^{+}(\bbs)$, the process $S^{\delta_{n_k}} Z$ is a nonnegative supermartingale, and hence by Doob's optional stopping theorem we obtain
\begin{align*}
\bbe[\xi Z_T] \leq \bbe \big[ S_T^{\delta_{n_k}} Z_T \big] \leq \bbe \big[ S_0^{\delta_{n_k}} Z_0 \big] = S_0^{\delta_{n_k}} Z_0.
\end{align*}
Since $S_0^{\delta_{n_k}} \downarrow 0$, we deduce that $\bbe[\xi Z_T] = 0$. Therefore, and since $\xi \geq 0$ and $\bbp(Z_T > 0) = 1$, we obtain the contradiction $\xi = 0$.
\end{proof}

\begin{proposition}\label{prop-ESMD-NA-2b}
Suppose that an ESMD $Z$ for $\bbp_{\sfi}^{\delta \geq 0}(\bbs) \cap \bbp_{\sfi}^{+}(\bbs)$ exists. Let
\begin{align*}
(S^{\delta_n})_{n \in \bbn} \subset \bbp_{\sfi}(\bbs) 
\end{align*}
be an asymptotic arbitrage of the first kind. Then there is a subsequence $(S^{\delta_{n_k}})_{k \in \bbn}$ such that $S^{\delta_{n_k}} \notin \bbp_{\sfi}^{\delta \geq 0}(\bbs)$ for all $k \in \bbn$.
\end{proposition}

\begin{proof}
We have $S_0^{\delta_n} \downarrow 0$ and there exists a sequence of random variables $(\xi_n)_{n \in \bbn} \subset L_+^0$ such that $S_T^{\delta_n} \geq \xi_n$ for all $n \in \bbn$ and
\begin{align}\label{limsup-positive}
\limsup_{n \to \infty} \bbp(\xi_n \geq 1) > 0. 
\end{align}
Suppose, contrary to the assertion above, there is an index $n_0 \in \bbn$ such that $S^{\delta_n} \in \bbp_{\sfi}^{\delta \geq 0}(\bbs)$ for all $n \geq n_0$. According to Lemma \ref{lemma-neg-port} we have $S^{\delta_n} \in \bbp_{\sfi}^+(\bbs)$ for all $n \geq n_0$. By (\ref{limsup-positive}) there are $\epsilon \in (0,1]$ and a subsequence $(n_k)_{k \in \bbn}$ with $n_1 \geq n_0$ such that
\begin{align}\label{prob-xi-delta}
\bbp(\xi_{n_k} \geq 1) \geq \epsilon \quad \text{for all $k \in \bbn$.}
\end{align}
By the von Weizs\"{a}cker theorem (see \cite{Weizsaecker}) there exist another subsequence $(\xi_{n_{k_l}})_{l \in \bbn}$ and a nonnegative random variable $\xi : \Omega \to [0,\infty]$ such that $(\xi_{n_{k_l}})_{l \in \bbn}$ is $\bbp$-almost surely Ces\`{a}ro convergent to $\xi$. In other words, we have $\bar{\xi}_{n_{k_l}} \overset{\text{a.s.}}{\to} \xi$, where
\begin{align*}
\bar{\xi}_{n_{k_l}} := \frac{1}{l} \sum_{p=1}^l \xi_{n_{k_p}} \quad \text{for each $l \in \bbn$.}
\end{align*}
Noting (\ref{prob-xi-delta}), according to \cite[Lemma 9.8.6]{DS-book} (applied with $\eta = 1/2$) we have
\begin{align*}
\bbp \bigg( \sum_{p=1}^l \xi_{n_{k_p}} \geq \frac{\epsilon l}{2} \bigg) \geq \frac{\epsilon / 2}{1 - \epsilon/2} = \frac{\epsilon}{2 - \epsilon} \quad \text{for each $l \in \bbn$,}
\end{align*}
which implies
\begin{align*}
\bbp \bigg( \bar{\xi}_{n_{k_l}} \geq \frac{\epsilon}{2} \bigg) \geq \frac{\epsilon}{2 - \epsilon} \quad \text{for each $l \in \bbn$.}
\end{align*}
Therefore, and since $\bar{\xi}_{n_{k_l}} \overset{\bbp}{\to} \xi$, we have $\xi \neq 0$. Now, let $l \in \bbn$ be arbitrary. Since $Z$ is an ESMD for $\bbp_{\sfi}^{\delta \geq 0}(\bbs) \cap \bbp_{\sfi}^{+}(\bbs)$, the processes $S^{\delta_{n_{k_p}}} Z$, $p=1,\ldots,l$ are nonnegative supermartingales, and hence by Doob's optional stopping theorem we obtain
\begin{align*}
\bbe \big[ \bar{\xi}_{n_{k_l}} Z_T \big] = \frac{1}{l} \sum_{p=1}^l \bbe \big[ \xi_{n_{k_p}} Z_T \big] \leq \frac{1}{l} \sum_{p=1}^l \bbe \Big[ S_T^{\delta_{n_{k_p}}} Z_T \Big] \leq \frac{1}{l} \sum_{k=1}^l \bbe \Big[ S_0^{\delta_{n_{k_p}}} Z_0 \Big] = Z_0 \sigma_l,
\end{align*}
where $(\sigma_l)_{l \in \bbn} \subset (0,\infty)$ denotes the sequence
\begin{align*}
\sigma_l := \frac{1}{l} \sum_{p=1}^l S_0^{\delta_{n_{k_p}}}, \quad l \in \bbn.
\end{align*}
Since $S_0^{\delta_{n_{k_l}}} \to 0$, we have $\sigma_l \to 0$. Therefore, by Fatou's lemma we deduce that
\begin{align*}
\bbe[\xi Z_T] = \bbe \Big[ \lim_{l \to \infty} \bar{\xi}_{n_{k_l}} Z_T \Big] \leq \liminf_{l \to \infty} \bbe \big[ \bar{\xi}_{n_{k_l}} Z_T \big] \leq 0,
\end{align*}
and hence $\bbe[\xi Z_T] = 0$. Since $\xi \geq 0$ and $\bbp(Z_T > 0) = 1$, we obtain the contradiction $\xi = 0$.
\end{proof}

\begin{proposition}\label{prop-ESMD-NA-2b-v2}
Suppose that an ESMD $Z$ for $\bbp_{\sfi}^{\delta \geq 0}(\bbs) \cap \bbp_{\sfi}^{+}(\bbs)$ exists. Let
\begin{align*}
(S^{\delta_n})_{n \in \bbn} \subset \bbp_{\sfi}(\bbs) 
\end{align*}
be a limit arbitrage of the first kind. Then there is an index $n_0 \in \bbn$ such that $S^{\delta_n} \notin \bbp_{\sfi}^{\delta \geq 0}(\bbs)$ for all $n \geq n_0$.
\end{proposition}

\begin{proof}
We have $S_0^{\delta_n} \downarrow 0$ and there exists a sequence of random variables $(\xi_n)_{n \in \bbn} \subset L_+^0$ such that $S_T^{\delta_n} \geq \xi_n$ for all $n \in \bbn$ and
\begin{align}\label{lim-positive}
\lim_{n \to \infty} \bbp(\xi_n \geq 1) > 0. 
\end{align}
Suppose, contrary to the assertion above, there is a subsequence $(S^{\delta_{n_k}})_{k \in \bbn}$ such that $S^{\delta_{n_k}} \in \bbp_{\sfi}^{\delta \geq 0}(\bbs)$ for all $k \in \bbn$. According to Lemma \ref{lemma-neg-port} we have $S^{\delta_{n_k}} \in \bbp_{\sfi}^+(\bbs)$ for all $k \in \bbn$. By (\ref{lim-positive}) there are $\epsilon \in (0,1]$ and an index $n_0 \in \bbn$ such that
\begin{align}\label{prob-xi-delta-2}
\bbp(\xi_n \geq 1) \geq \epsilon \quad \text{for all $n \geq n_0$.}
\end{align}
Without loss of generality, we may assume that $n_0 \leq n_1$. By the von Weizs\"{a}cker theorem (see \cite{Weizsaecker}) there exist another subsequence $(\xi_{n_{k_l}})_{l \in \bbn}$ and a nonnegative random variable $\xi : \Omega \to [0,\infty]$ such that $(\xi_{n_{k_l}})_{l \in \bbn}$ is $\bbp$-almost surely Ces\`{a}ro convergent to $\xi$. Now, as in the proof of Proposition \ref{prop-ESMD-NA-2b} we derive the contradiction that simultaneously $\xi \neq 0$ and $\xi = 0$.
\end{proof}

\begin{proposition}\label{prop-ESMD-NA-3}
Suppose that an ESMD $Z$ for $\bbp_{\sfi}^{\delta \geq 0}(\bbs) \cap \bbp_{\sfi}^{+}(\bbs)$ exists. Let
\begin{align*}
(S^{\delta_n})_{n \in \bbn} \subset \bbp_{\sfi}(\bbs)
\end{align*}
be a free lunch with vanishing risk. Then there exists an index $n_0 \in \bbn$ such that $S^{\delta_n} \notin \bbp_{\sfi}^{\delta \geq 0}(\bbs)$ for all $n \geq n_0$.
\end{proposition}

\begin{proof}
There exist a random variable $\xi \in L_+^{\infty} \setminus \{ 0 \}$ and a sequence $(\xi_n)_{n \in \bbn} \subset L^{\infty}$ such that $\| \xi_n - \xi \|_{L^{\infty}} \to 0$ and $S_T^{\delta_n} \geq \xi_n$ for each $n \in \bbn$. Since $\xi \geq 0$ and the mapping $x \mapsto x^+$ is Lipschitz continuous, we also have $\| \xi_n^+ - \xi \|_{L^{\infty}} \to 0$, and hence $\| \xi_n^+ \|_{L^{\infty}} \to \| \xi \|_{L^{\infty}}$. Moreover, since $\| \xi \|_{L^{\infty}} > 0$, there exists an index $n_0 \in \bbn$ such that $\| \xi_n^+ \|_{L^{\infty}} > 0$ for each $n \geq n_0$. Now, let $n \geq n_0$ be arbitrary. We claim that $S^{\delta_n} \notin \bbp_{\sfi}^{\delta \geq 0}(\bbs)$. Indeed, by Lemma \ref{lemma-neg-port} we may assume that $S^{\delta_n} \in \bbp_{\sfi}^+(\bbs)$. Thus, since $S_T^{\delta_n} \geq 0$, we have $S_T^{\delta_n} \geq \xi_n^+$. Furthermore, since $S_0^{\delta_n} = 0$ and $\xi_n^+ \in L_+^{\infty} \setminus \{ 0 \}$, the portfolio $S^{\delta_n}$ is an arbitrage portfolio. Therefore, by Proposition \ref{prop-ESMD-NA} we obtain $S^{\delta_n} \notin \bbp_{\sfi}^{\delta \geq 0}(\bbs)$.
\end{proof}

\begin{proposition}\label{prop-ESMD-NA-4}
Suppose that an ESMD $Z$ for $\bbp_{\sfi}^{\delta \geq 0}(\bbs) \cap \bbp_{\sfi}^{+}(\bbs)$ with $Z_T \in L^1$ exists. Let
\begin{align*}
(S^{\delta_j})_{j \in J} \subset \bbp_{\sfi}(\bbs)
\end{align*}
be a free lunch. Then there exists an index $j_0 \in J$ such that $S^{\delta_j} \notin \bbp_{\sfi}^{\delta \geq 0}(\bbs)$ for all $j \geq j_0$.
\end{proposition}

\begin{proof}
There exist a random variable $\xi \in L_+^{\infty} \setminus \{ 0 \}$ and a net $(\xi_j)_{j \in J} \subset L^{\infty}$ such that $\xi_j \overset{w^*}{\to} \xi$ and $S_T^{\delta_j} \geq \xi_j$ for each $j \in J$. Therefore, we have $\bbe[\xi_j \eta] \to \bbe[\xi \eta]$ for all $\eta \in L^1$. In particular, since $Z_T \in L^1$, we obtain
\begin{align}\label{xi-Z-conv}
\bbe[\xi_j Z_T] \to \bbe[\xi Z_T].
\end{align}
Suppose, contrary to the assertion above, that for every $j_0 \in J$ there exists $j \geq j_0$ such that $S^{\delta_j} \in \bbp_{\sfi}^{\delta \geq 0}(\bbs)$. According to Lemma \ref{lemma-neg-port} we have $S^{\delta_j} \in \bbp_{\sfi}^+(\bbs)$. Moreover, we have
\begin{align}\label{E-xi-j-zero}
\bbe[\xi_j Z_T] = 0.
\end{align}
Indeed, since $Z$ is an ESMD for $\bbp_{\sfi}^{\delta \geq 0}(\bbs) \cap \bbp_{\sfi}^{+}(\bbs)$, the process $S^{\delta_j} Z$ is a nonnegative supermartingale, and hence, recalling that $S_0^{\delta_j} = 0$, by Doob's optional stopping theorem we obtain
\begin{align*}
\bbe[\xi_j Z_T] \leq \bbe[S_T^{\delta_j} Z_T] \leq \bbe [S_0^{\delta_j} Z_0] = S_0^{\delta_j} Z_0 = 0.
\end{align*}
Now, we will show that $\bbe[\xi Z_T] = 0$. Indeed, let $\epsilon > 0$ be arbitrary. Due to (\ref{xi-Z-conv}), there exists $j_0 \in J$ such that
\begin{align*}
| \bbe[\xi_j Z_T] - \bbe[\xi Z_T] | \leq \epsilon \quad \text{for all $j \geq j_0$.}
\end{align*}
Thus, using (\ref{E-xi-j-zero}) it follows that
\begin{align*}
\bbe[\xi Z_T] \leq \epsilon.
\end{align*}
Since $\epsilon > 0$ was arbitrary, we deduce that $\bbe[\xi Z_T] = 0$. Now, noting that $\xi \geq 0$ and $\bbp(Z_T > 0) = 1$, we obtain the contradiction $\xi = 0$.
\end{proof}

Now, our goal is to prove that an ESMD $Z$ for $\bbp_{\sfi}^{\delta \geq 0}(\bbs) \cap \bbp_{\sfi}^{+}(\bbs)$ exists. Recall that due to Lemmas \ref{lemma-fin-market-1} and \ref{lemma-fin-market-2} we may assume that we are in the situation of Remark \ref{rem-fin-market}. We define $X^i := S^i / S^d$ for $i=1,\ldots,d-1$ and the $\bbr^d$-valued semimartingale $\bar{X} := (X,1)$, where $X := (X^1,\ldots,X^{d-1})$. Furthermore, we define $\bbx := \{ X^1,\ldots,X^{d-1} \}$ and the discounted market $\bar{\bbx} := \{ \bar{X}^1,\ldots,\bar{X}^d \}$. The following are two well-known results about the change of num\'{e}raire technique.

\begin{lemma}\cite[Prop. 5.2]{Takaoka-Schweizer}\label{lemma-change-num-1}
For an $\bbr^d$-valued predictable process $\delta$ the following statements are equivalent:
\begin{enumerate}
\item[(i)] We have $\delta \in \Delta_{\sfi}(\bbs)$.

\item[(ii)] We have $\delta \in \Delta_{\sfi}(\bar{\bbx})$.
\end{enumerate}
If the previous conditions are fulfilled, then we have
\begin{align*}
S^{\delta} = S^d \bar{X}^{\delta}.
\end{align*}
\end{lemma}

For a pair $(x,\theta) \in \bbr \times \Delta(\bbx)$ the corresponding wealth process is defined as
\begin{align*}
X^{x,\theta} := x + \theta \bdot X.
\end{align*}

\begin{lemma}\cite[Lemma 5.1]{Takaoka-Schweizer}\label{lemma-change-num-2}
There is a bijection between $\bbr \times \Delta(\bbx)$ and $\Delta_{\sfi}(\bar{\bbx})$, which is defined as follows:
\begin{enumerate}
\item For $\delta \in \Delta_{\sfi}(\bar{\bbx})$ we assign
\begin{align*}
\delta \mapsto (x,\theta) := (X_0^{\delta},\delta^1,\ldots,\delta^{d-1}) \in \bbr \times \Delta(\bbx).
\end{align*}
\item For $(x,\theta) \in \bbr \times \Delta(\bbx)$ we assign
\begin{align*}
(x,\theta) \mapsto \delta = (\theta, X_-^{x,\theta} - \theta \cdot X_- ) \in \Delta_{\sfi}(\bar{\bbx}).
\end{align*}
\end{enumerate}
Furthermore, for all $(x,\theta) \in \bbr \times \Delta(\bbx)$ and the corresponding self-financing strategy $\delta \in \Delta_{\sfi}(\bar{\bbx})$ we have
\begin{align*}
\bar{X}^{\delta} = X^{x,\theta}.
\end{align*}
\end{lemma}

\begin{proposition}\label{prop-ESMD-exists}
Let $\bbs = \{ S^1,\ldots,S^d \}$ be a financial market consisting of nonnegative semimartingales for some $d \in \bbn$ with $d \geq 2$, such that the semimartingales $S^1,\ldots,S^{d-1} \geq 0$ cannot revive from bankruptcy, and we have $S^d, S_-^d > 0$. Then the following statements are true:
\begin{enumerate}
\item There exists an ESMD $Z$ for $\bbp_{\sfi}^{\delta \geq 0}(\bbs) \cap \bbp_{\sfi}^{+}(\bbs)$.

\item If $(S_T^d)^{-1} \in L^{\infty}$, then the ESMD $Z$ can be chosen such that $Z_T \in L^1$.
\end{enumerate}
\end{proposition}

\begin{proof}
According to \cite[Thm. 1.3]{Kardaras-Platen} there exists an ESMD $Y$ for $\bar{\bbx}$ such that $Y_0 = 1$. We set
\begin{align}\label{def-Z}
Z := Y / S^d. 
\end{align}
Since $1 \in \bar{\bbx}$, the process $Y$ is a supermartingale. Thus, if $(S_T^d)^{-1} \in L^{\infty}$, then we have $Z_T \in L^1$. Now, let $S^{\delta} \in \bbp_{\sfi}^{\delta \geq 0}(\bbs) \cap \bbp_{\sfi}^{+}(\bbs)$ be arbitrary. By Lemma \ref{lemma-change-num-1} we have $\bar{X}^{\delta} \in \bbp_{\sfi}^{\delta \geq 0}(\bar{\bbx}) \cap \bbp_{\sfi}^{+}(\bar{\bbx})$ and the identity
\begin{align}\label{change-num-2}
S^{\delta} = S^d \bar{X}^{\delta}.
\end{align}
We define the increasing sequence $(D_n)_{n \in \bbn}$ of predictable sets by $D_n := \{ \| \delta \| \leq n \}$ for all $n \in \bbn$, and the sequence of predictable, bounded, nonnegative processes $(\delta_n)_{n \in \bbn}$ as $\delta_n := \delta \bbI_{D_n}$ for all $n \in \bbn$. Let $n \in \bbn$ be arbitrary. By Lemma \ref{lemma-change-num-2} we have
\begin{align*}
\bar{X}^{\delta_n} = X^{x_n,\theta_n}.
\end{align*}
where $(x_n,\theta_n) \in \bbr \times \Delta(\bbx)$ is given by
\begin{align*}
(x_n,\theta_n) := (X_0^{\delta_n},\delta_n^1,\ldots,\delta_n^{d-1}),
\end{align*}
and we have
\begin{align*}
\delta_n = (\theta_n, X_-^{x_n,\theta_n} - \theta_n \cdot X_- ).
\end{align*}
Using integration by parts, we obtain 
\begin{align*}
Y \bar{X}^{\delta_n} = Y X^{x_n,\theta_n} &= x_n + \theta_n \bdot (Y X) + (X_-^{x_n,\theta_n} - \theta_n \cdot X_-) \bdot Y
\\ &= x_n + \delta_n \bdot (Y \bar{X}),
\end{align*}
where $Y \bar{X}$ denotes the $\bbr^d$-valued supermartingale with components $Y \bar{X}^i$, $i=1,\ldots,d$; cf. \cite[p. 2685]{Kardaras-Platen}. By \cite[Prop. 5.2]{Takaoka-Schweizer} we have $\delta \in \Delta_{\sfi}(Y \bar{\bbx})$, where $Y \bar{\bbx} := \{ Y \bar{X}^1, \ldots, Y \bar{X}^d \}$. Hence, using \cite[Thm. III.6.19.a]{Jacod-Shiryaev}, letting $n \to \infty$ we obtain
\begin{align*}
Y \bar{X}^{\delta} = x + \delta \bdot (Y \bar{X}).
\end{align*}
Now, for each $n \in \bbn$ we obtain
\begin{align*}
\bbI_{D_n} \bdot (Y \bar{X}^{\delta}) = \bbI_{D_n} \bdot \big( \delta \bdot (Y \bar{X}) \big) = (\delta \bbI_{D_n}) \bdot (Y \bar{X}) = \delta_n \bdot (Y \bar{X}),
\end{align*}
where we have used the associativity of the stochastic integral; see \cite[Thm. 4.6]{Shiryaev-Cherny}. It follows that $\bbI_{D_n} \bdot (Y \bar{X}^{\delta})$ is a supermartingale for each $n \in \bbn$. Hence, we deduce that $Y \bar{X}^{\delta}$ is a nonnegative $\sigma$-supermartingale in the sense of \cite{Kallsen}, and by \cite[Prop. 3.1]{Kallsen} it follows that $Y \bar{X}^{\delta}$ is a supermartingale. Taking into account (\ref{def-Z}) and (\ref{change-num-2}), this shows that $Z S^{\delta}$ is a supermartingale. Consequently, the process $Z$ is an ESMD for $\bbp_{\sfi}^{\delta \geq 0}(\bbs) \cap \bbp_{\sfi}^{+}(\bbs)$.
\end{proof}

Now, we are ready to provide the proofs of the results stated in Section \ref{sec-intro}. Recall that due to Lemmas \ref{lemma-fin-market-1} and \ref{lemma-fin-market-2} we may assume that we are in the situation of Remark \ref{rem-fin-market}. The proofs of Theorems \ref{thm-main}--\ref{thm-main-3} and \ref{thm-main-3c} are a consequence of combining Propositions \ref{prop-ESMD-NA}--\ref{prop-ESMD-NA-4} with Proposition \ref{prop-ESMD-exists}, and the proof of Theorem \ref{thm-main-3b} is an immediate consequence of the more general Theorem \ref{thm-main-3c}.

\section{Financial market models in discrete time}\label{sec-discrete}

In this section we briefly demonstrate how the proof of our main result can be simplified for discrete time models with strictly positive primary security accounts. We assume that a discrete filtration $(\calf_t)_{t=0,\ldots,T}$ for some $T \in \bbn$ with $\calf_0 = \{ \Omega,\emptyset \}$ is given, and we consider a market $\bbs = \{ S^i : i \in I \}$ consisting of nonnegative, adapted processes $S^i \geq 0$ for all $i \in I$. As shown in \cite[page 14]{Jacod-Shiryaev}, this can be regarded as a particular case of the continuous time setting, which we have considered so far. In this particular setting, the space $\Delta_{\sfi}(\bbs)$ of all self-financing strategies consists of all predictable processes $\delta = (\delta_t)_{t=1,\ldots,T}$ with finite support (\ref{support}) such that
\begin{align}\label{discrete-self-fin}
\delta_t \cdot S_t = \delta_{t+1} \cdot S_t \quad \text{for all $t=1,\ldots,T-1$,}
\end{align}
and the corresponding self-financing portfolio $S^{\delta}$ is given by
\begin{align}\label{discrete-port-1}
S_0^{\delta} &= \delta_1 \cdot S_0,
\\ \label{discrete-port-2} S_t^{\delta} &= \delta_t \cdot S_t, \quad t=1,\ldots,T.
\end{align}
If the primary security accounts are strictly positive, then we can provide an elementary proof of Theorem \ref{thm-main} for such discrete time models by using the following result.

\begin{proposition}\label{prop-discrete}
If $S^i > 0$ for all $i \in I$, then for every self-financing portfolio $S^{\delta} \in \bbp_{\sfi}^{\delta \geq 0}(\bbs)$ with $S_0^{\delta} = 0$ we have $\delta = 0$, and hence $S^{\delta} = 0$.
\end{proposition}

\begin{proof}
Inductively we will show that $\delta_t = 0$ for all $t=1,\ldots,T$. For $t=1$ note that by (\ref{discrete-port-1}) we have $\delta_1 \cdot S_0 = 0$. Since $\delta_1^i \geq 0$ and $S_0^i > 0$ for all $i \in I$, we obtain $\delta_1 = 0$. We proceed with the induction step $t \to t+1$ and suppose that $\delta_t = 0$. Then by (\ref{discrete-self-fin}) we have $\delta_{t+1} \cdot S_t = 0$. Since $\delta_{t+1}^i \geq 0$ and $S_t^i > 0$ for all $i \in I$, we obtain $\delta_{t+1} = 0$, completing the proof.
\end{proof}

\section{Examples}\label{sec-examples}

In this section we present examples, where arbitrage portfolios can be explicitly computed, and verify that, in accordance with Theorem \ref{thm-main}, these portfolios indeed require short selling.

\begin{example}
Consider the market $\bbs = \{ S^1, S^2 \}$ consisting of two primary security accounts such that $S_0^1 = S_0^2$ and $S_t^1 > S_t^2$ for all $t > 0$. For example, the assets could be $S_t^1 = e^t$ for $t \in \bbr_+$ and $S^2 = 1$. Then $\delta = (1,-1)$ is a self-financing strategy, and we obtain the nonnegative arbitrage portfolio $S^{\delta} = S^1 - S^2$. In accordance with Theorem \ref{thm-main}, we see that this arbitrage portfolio requires short selling.
\end{example}

In the next example we consider a simple one-period model in discrete time, which is a particular case of the setting from Section \ref{sec-discrete}.

\begin{example}
This example can be found in \cite[Exercise 1.1.1]{FS}. The probability space $(\Omega,\calf,\bbp)$ is defined as follows. Let $\Omega := \{ \omega_1,\omega_2,\omega_3 \}$ with pairwise different elements $\omega_1, \omega_2, \omega_3$, the $\sigma$-algebra is given by the power set $\calf := \mathfrak{P}(\Omega)$, and concerning the probability measure $\bbp$ we assume that $\bbp(\{ \omega_i \}) > 0$ for all $i=1,2,3$. The filtration $(\calf_t)_{t=0,1}$ is given by $\calf_0 := \{ \Omega,\emptyset \}$ and $\calf_1 := \calf$. In this example we have three primary security accounts. The values at time $t=0$ are given by the vector $\pi \in \bbr^3$ defined as
\begin{align*}
\pi := \left(
\begin{array}{c}
1
\\ 2
\\ 7
\end{array}
\right),
\end{align*}
and the values at the terminal time $T=1$ are given by the random vector $S : \Omega \to \bbr^3$ defined as
\begin{align*}
S(\omega_1) := \left(
\begin{array}{c}
1
\\ 3
\\ 9
\end{array}
\right), \quad S(\omega_2) := \left(
\begin{array}{c}
1
\\ 1
\\ 5
\end{array}
\right) \quad \text{and} \quad S(\omega_3) := \left(
\begin{array}{c}
1
\\ 5
\\ 10
\end{array}
\right).
\end{align*}
Hence, the market $\bbs = \{ S^1,S^2,S^3 \}$ consists of one savings account and two risky security accounts. Note that in this example every arbitrage portfolio $S^{\delta} \in \bbp_{\sfi}(\bbs)$ is nonnegative; that is, we have $S^{\delta} \in \bbp_{\sfi}^+(\bbs)$. It is easily verified that the market $\bbs$ admits arbitrage, and that all arbitrage portfolios are specified by the set
\begin{align*}
\left\{ \left(
\begin{array}{c}
3 \xi
\\ 2 \xi
\\ -\xi
\end{array}
\right) : \xi > 0 \right\}.
\end{align*}
Hence, in accordance with Theorem \ref{thm-main}, each such arbitrage portfolio requires short selling. More precisely, in order to exploit arbitrage, we have to go short in the second risky primary security account, whereas for the first risky primary security account and for the savings account we have long positions.
\end{example}

The above two examples we have considered financial markets, where NA is not satisfied; that is, at least one nonnegative arbitrage portfolio exists. Now, we are interested in markets satisfying NA, and still admitting arbitrage portfolios which can go negative, e.g. admissible arbitrage portfolios, and to confirm the validity of Theorem \ref{thm-main} also for such arbitrage portfolios.

For what follows, we consider a finite market $\bbs = \{ S^1,\ldots,S^d \}$ consisting of nonnegative semimartingales. Suppose the model satisfies NA; that is, a nonnegative arbitrage portfolio does not exist. We are interested in a criterion which ensures the existence of an arbitrage portfolio. As we will see below, such a criterion is the existence of a money market bubble. Important references about bubbles include, for example, \cite{HLW, Cox-Hobson, Jarrow-Madan, JPS, Guasoni} and the article \cite{Protter-bubbles} with an overview of the respective literature.

Before we proceed, let us clarify the relation between our upcoming results and existing results in the bubble literature. Usually, when introducing bubbles, one assumes that an equivalent risk-neutral probability measure $\bbq \approx \bbp$ exists, which is equivalent to no-arbitrage in the sense of NFLVR. Then a price process $S^i$ has a bubble if it is a strict supermartingale under the risk-neutral probability measure $\bbq$. For what follows, we will drop the typical assumption that an equivalent risk-neutral probability measure exists, which means NFLVR, and only assume that NUPBR holds true.

\begin{definition}
The market $\bbs$ has a \emph{bubble} if there are an index $i \in \{ 1,\ldots,d \}$ and a self-financing portfolio $S^{\vartheta} \in \bbp_{\sfi}^+(\bbs)$ such that $S_0^{\vartheta} < S_0^i$ and $S_T^{\vartheta} \geq S_T^i$. In this case, the bubble is $S^i$.
\end{definition}

\begin{remark}
Since we assume NUPBR rather than NFLVR, our definition of a bubble only relates to nonnegative, rather than admissible portfolios.
\end{remark}

\begin{lemma}\label{lemma-bubble-arb}
Suppose $S^i,S_-^i > 0$ for all $i=1,\ldots,d$, and that the market $\bbs$ has a bubble. Then there exists an arbitrage portfolio $S^{\delta}$.
\end{lemma}

\begin{proof}
There are an index $i \in \{ 1,\ldots,d \}$ and a self-financing portfolio $S^{\vartheta} \in \bbp_{\sfi}^+(\bbs)$ such that $S_0^{\vartheta} < S_0^i$ and $S_T^{\vartheta} \geq S_T^i$. If $S_0^{\vartheta} = 0$, then $S^{\delta}$ with $\delta := \vartheta$ is already an arbitrage portfolio. Now, we assume that $S_0^{\vartheta} > 0$. Setting
\begin{align*}
\alpha := \frac{S_0^i}{S_0^{\vartheta}} > 1,
\end{align*}
by \cite[Lemma 7.21]{Platen-Tappe-tvs} the process $\delta := \alpha \vartheta - e_i$ is also a self-financing strategy for the market $\bbs$, and we have
\begin{align}\label{arb-port-bubble}
S^{\delta} = \alpha S^{\vartheta} - S^i.
\end{align}
It follows that $S_0^{\delta} = 0$ and $S_T^{\delta} > 0$, showing that $S^{\delta}$ is an arbitrage portfolio.
\end{proof}

\begin{remark}
As we can see from (\ref{arb-port-bubble}), the arbitrage portfolio $S^{\delta}$ is generally not admissible; in particular, it can happen that it goes negative.
\end{remark}

For what follows, we recall that $\calm$ denotes the space of all uniformly integrable martingales, and that $\calm_{\loc}$ denotes the space of all local martingales. A semimartingale $Z$ with $Z,Z_- > 0$ is called an \emph{equivalent local martingale deflator} (ELMD) for $\bbs$ if $S^i Z \in \calm_{\loc}$ for each $i=1,\ldots,d$.

\begin{lemma}\label{lemma-bubble}
Suppose there exists an ELMD $Z$ for $\bbs$ such that $S^i Z \in \calm$ for each $i=1,\ldots,d$. Then the market $\bbs$ does not have a bubble.
\end{lemma}

\begin{proof}
According to \cite[Prop. C.4]{Platen-Tappe-FTAP} the process $Z$ is also an ELMD for $\bbp_{\sfi}^+(\bbs)$. Let $i \in \{ 1,\ldots,d \}$ and $S^{\vartheta} \in \bbp_{\sfi}^+(\bbs)$ be such that $S_T^{\vartheta} \geq S_T^i$. By \cite[Lemma I.1.44]{Jacod-Shiryaev} and Doob's optional stopping theorem for nonnegative supermartingales we obtain
\begin{align*}
(S_0^{\vartheta} - S_0^i)Z_0 &= \bbe[S_0^{\vartheta}Z_0] -\bbe[S_0^i Z_0]
\\ &\geq \bbe[S_T^{\vartheta} Z_T] - \bbe[S_T^i Z_T] = \bbe[(S_T^{\vartheta} - S_T^i) Z_T] \geq 0,
\end{align*}
and hence $S_0^{\vartheta} \geq S_0^i$.
\end{proof}

The following result contains sufficient conditions for the converse statement of the previous result.

\begin{lemma}\label{lemma-bubble-suff}
Suppose there exists an ELMD $Z$ for $\bbs$ such that for some $i \in \{ 1,\ldots,d \}$ we have $S^i Z \notin \calm$. Furthermore, we assume that $S_T^i Z_T \in L^1$ and that there exists a self-financing portfolio $S^{\vartheta} \in \bbp_{\sfi}^+(\bbs)$ such that
\begin{align}\label{port-for-bubble}
S_t^{\vartheta} = \frac{\bbe[S_T^i Z_T | \calf_t]}{Z_t}, \quad t \in [0,T].
\end{align}
Then the market $\bbs$ has the bubble $S^i$.
\end{lemma}

\begin{proof}
By (\ref{port-for-bubble}) we have $S^{\vartheta} Z \in \calm$ and $S_T^{\vartheta} = S_T^i$. Furthermore, since $S^i Z \notin \calm$, we have $\bbe[S_0^i Z_0] > \bbe[S_T^i Z_T]$. Therefore, we obtain
\begin{align*}
S_0^{\vartheta} Z_0 = \bbe[ S_T^{\vartheta} Z_T ] = \bbe[S_T^i Z_T] < S_0^i Z_0,
\end{align*}
and hence $S_0^{\vartheta} < S_0^i$.
\end{proof}

We call every predictable, c\`{a}dl\`{a}g, finite variation process $B$ with $B_0 = 1$ and $B,B_- > 0$ a \emph{savings account}.

\begin{proposition}\label{prop-FTAP-bubbles}
Suppose that $S^d = B$ for some savings account $B$ such that $B$ and $B^{-1}$ are bounded. Furthermore, we assume there exists a local martingale $D > 0$ with $D_0 = 1$ such that the multiplicative special semimartingale $Z = D B^{-1}$ is an ELMD for $\bbs$. Then the following statements are true:
\begin{enumerate}
\item Nonnegative arbitrage portfolios, nonnegative (asymptotic) arbitrages of the first kind and nonnegative free lunches (with vanishing/bounded risk) do not exist.

\item If $D \in \calm$, then an admissible free lunch with vanishing risk does not exist. In particular, an admissible arbitrage portfolio does not exist.

\item If $D$ is the unique local martingale with $D_0 = 1$ such that $D B^{-1}$ is an ELMD for $\bbs$, then the existence of an admissible arbitrage portfolio (or, more generally, of an admissible free lunch with vanishing risk) implies that $D \notin \calm$.
\end{enumerate}
\end{proposition}

\begin{proof}
These statements follow from \cite[Thm. 4.3 and Thm. 4.6]{Platen-Tappe-FTAP}.
\end{proof}

Recall that we are interested in market models satisfying NA for its nonnegative self-financing portfolios, but still admitting an arbitrage portfolio. The following remark summarizes our previous results.

\begin{remark}\label{rem-connection}
Suppose we are in the situation of Proposition \ref{prop-FTAP-bubbles}. Then the market $\bbs$ satisfies NUPBR, and hence, a nonnegative arbitrage portfolio does not exist. Furthermore, the local martingale $D$ is a supermartingale and a candidate for the density process of an equivalent risk-neutral probability measure $\bbq \approx \bbp$ for the discounted market $\{ S^1 / B, \ldots, S^{d-1} / B \}$, which exists if and only if $D$ is a true martingale. Consider the following three statements:
\begin{enumerate}
\item[(a)] There exists an admissible arbitrage portfolio.

\item[(b)] There exists an arbitrage portfolio.

\item[(c)] The market has a bubble.

\item[(d)] $D$ is a strict local martingale.
\end{enumerate}
Then the following statements are true:
\begin{enumerate}
\item We have the implications (a) $\Rightarrow$ (b) \& (d).

\item If $S^i,S_-^i > 0$ for all $i=1,\ldots,d-1$, then we have the implication (c) $\Rightarrow$ (b).

\item If $D$ is the unique local martingale with $D_0 = 1$ such that $D B^{-1}$ is an ELMD for $\bbs$, then we have the implication (d) $\Rightarrow$ (a).

\item If $S^1 Z, \ldots, S^{d-1} Z \in \calm$, then we have the implication (c) $\Rightarrow$ (d).

\item If there exists a self-financing portfolio $S^{\vartheta} \in \bbp_{\sfi}^+(\bbs)$ such that
\begin{align*}
S_t^{\vartheta} = \frac{\bbe[D_T | \calf_t]}{Z_t}, \quad t \in [0,T],
\end{align*}
then we have the implication (d) $\Rightarrow$ (c).

\end{enumerate}
\end{remark}

In the upcoming result, we have situations where nonnegative arbitrage portfolios do not exist and all the statements (a)--(d) hold.

\begin{proposition}\label{prop-strict-local}
Let $W$ be an $\bbr$-valued Wiener process, and assume that the filtration $(\calf_t)_{t \in [0,T]}$ is the right-continuous, completed filtration generated by $W$. Suppose that the market is of the form $\bbs = \{ S,1 \}$, containing a continuous semimartingale $S > 0$ such that $S^{-1} \in \calm_{\loc}$ with $S^{-1} \notin \calm$ and the quadratic variation $\la S^{-1} \ra$ are strictly increasing. Then the following statements are true:
\begin{enumerate}
\item A nonnegative arbitrage portfolio does not exist.

\item There exists an admissible arbitrage portfolio.

\item The market has a bubble.
\end{enumerate}
\end{proposition}

\begin{proof}
Without loss of generality, we may assume that $S_0 = 1$. Setting $Z := D:= S^{-1}$, we have $Z \in \calm_{\loc}$. Thus $Z$ is an ELMD for $\bbs$, and hence, a nonnegative arbitrage portfolio does not exist.  

Now, let $Y$ be any ELMD for $\bbs$ with $Y_0 = 1$. By the martingale representation theorem (see \cite[Thm. III.4.33]{Jacod-Shiryaev}) there exist $\theta, \sigma \in L_{\loc}^2(W)$ and $A \in \calv$ such that $Y = \cale(-\theta \bdot W)$ and $S = \cale(A + \sigma \bdot W)$. Furthermore, by Yor's formula (see \cite[II.8.19]{Jacod-Shiryaev}) we obtain
\begin{align*}
YS = \cale \big( (\sigma - \theta) \bdot W + A - \sigma \theta \bdot \lambda \big),
\end{align*}
where $\lambda$ denotes the Lebesgue measure. Since $YS \in \calm_{\loc}$, there exists $a \in L_{\loc}^1(A)$ such that $A = a \bdot \lambda$ up to an evanescent set, and it follows that $\theta = a / \sigma$ $\lambda$-a.e. $\bbp$-a.e. Therefore, we have $Y = \cale(-(a / \sigma) \bdot W)$ up to an evanescent set. Since $Y$ was an arbitrary ELMD for $\bbs$ with $Y_0 = 1$, this proves the uniqueness of the ELMD $Z$. Hence, by Remark \ref{rem-connection} there exists an admissible arbitrage portfolio.

Since $Z$ is nonnegative, it is also a supermartingale, and hence we have $Z_T \in L^1$. We define the martingale $M \in \calm$ as
\begin{align*}
M_t := \bbe[Z_T | \calf_t], \quad t \in [0,T].
\end{align*}
By the martingale representation theorem (see \cite[Thm. III.4.33]{Jacod-Shiryaev}) there exist $H,K \in L_{\loc}^2(W)$ such that $Z = Z_0 + H \bdot W$ and $M = M_0 + K \bdot W$. Since the quadratic variation $\la Z \ra = H^2 \bdot \lambda$ is strictly increasing, we have $H \neq 0$ up to an evanescent set. We define the predictable process $\theta := \frac{K}{H}$. Since the Wiener process $W$ is continuous, we have $H, \theta H \in L_{\loc}^1(W)$. Therefore, by the first associativity theorem for vector stochastic integrals with respect to local martingales (proven as \cite[Thm. 4.6]{Shiryaev-Cherny} by using \cite[Lemma 4.4]{Shiryaev-Cherny}) we have $\theta \in L_{\loc}^1(Z)$ and
\begin{align*}
K \bdot W = (\theta H) \bdot W = \theta \bdot (H \bdot W) = \theta \bdot Z.
\end{align*}
Therefore, we have
\begin{align*}
M = M_0 + \theta \bdot Z.
\end{align*}
We define the $\bbr^2$-valued predictable process $\xi = (\eta,\theta)$ as
\begin{align*}
\eta := \theta \bdot Z - \theta Z.
\end{align*}
Introducing the discounted market $\bar{\bbs} = \{1,Z\}$, we have $\xi \in \Delta_{\sfi}(\bar{\bbs})$ and
\begin{align*}
M = M_0 + (1,Z)^{\xi} = M_0 + \eta + \theta Z.
\end{align*}
By \cite[Prop. 5.2]{Takaoka-Schweizer} we have $\xi \in \Delta_{\sfi}(\bbs)$. Furthermore, we obtain
\begin{align*}
MS = M_0 S + (M-M_0) S = M_0 S + (S,1)^{\xi} = M_0 S + \eta S + \theta.
\end{align*}
Now, we define the $\bbr^2$-valued predictable process $\vartheta$ as
\begin{align*}
\vartheta := \xi + M_0 e_1
\end{align*}
By \cite[Lemma 7.21]{Platen-Tappe-tvs} we have $\vartheta \in \Delta_{\sfi}(\bbs)$ and
\begin{align*}
(S,1)^{\vartheta} = (S,1)^{\xi} + M_0 S.
\end{align*}
Therefore, we have
\begin{align*}
(S,1)^{\vartheta} = MS,
\end{align*}
which means that
\begin{align*}
(S,1)_t^{\vartheta} = \frac{\bbe[Z_T | \calf_t]}{Z_t}, \quad t \in [0,T].
\end{align*}
Consequently, by Remark \ref{rem-connection} the market has a bubble.
\end{proof}

Let us mention well-known examples of a market $\bbs = \{ S,1 \}$, where $S$ is given as the power of a Bessel process.

\begin{example}
For a Bessel process $S > 0$ of dimension $\delta > 2$ its power $Z := S^{2-\delta}$, which is a solution to the SDE
\begin{align*}
dZ_t = (2-\delta)Z_t^{\frac{1-\delta}{2-\delta}} dW_t,
\end{align*}
is a strict local martingale; see \cite{Revuz-Yor}. Therefore, Proposition \ref{prop-strict-local} applies to the market $\bbs = \{ S^{\delta - 2}, 1 \}$.
\end{example}

Now, we provide an example of a financial market satisfying NA, and admitting an admissible arbitrage portfolio which can be computed explicitly.

\begin{example}\label{ex-3}
Consider the market $\bbs = \{ S,1 \}$, where $S$ is a Bessel process of dimension 3. By Proposition \ref{prop-strict-local} the following statements are true:
\begin{enumerate}
\item A nonnegative arbitrage portfolio does not exist.

\item There exists an admissible arbitrage portfolio.

\item The market has a bubble.
\end{enumerate}
The existence of arbitrage in this market model has been pointed out in \cite{DS-Bessel}, and explicit constructions of arbitrage portfolios are provided in \cite{Karatzas-Kardaras} and \cite{Ruf}. Assume that $S_0 = 1$ and $T=1$. Denoting by $\Phi : \bbr \to [0,1]$ the distribution function of the standard normal distribution ${\rm N}(0,1)$, we define the function $F : [0,1] \times \bbr \to [0,\frac{1}{\Phi(1)}]$ as
\begin{align*}
F(t,x) := \frac{\Phi(x / \sqrt{1-t})}{\Phi(1)},
\end{align*}
where for $t=1$ and $x$ positive this expression is understood to be $\frac{1}{\Phi(1)}$. Then, according to \cite[Example 4.6]{Karatzas-Kardaras}, an arbitrage portfolio $S^{\delta}$ is given by the self-financing strategy $\delta = (\theta,\eta)$ defined by the continuous processes
\begin{align}\label{strat-Bessel-1}
\theta_t &:= \frac{\partial}{\partial x} F(t,S_t), \quad t \in [0,1],
\\ \label{strat-Bessel-2} \eta &:= \theta \bdot S - \theta S,
\end{align}
and we have
\begin{align*}
S_t^{\delta} = F(t,S_t) - 1, \quad t \in [0,1].
\end{align*}
In particular, we see that this arbitrage portfolio is indeed admissible; that is, we have $S^{\delta} \in \bbp_{\sfi}^{\adm}(\bbs)$. However, this arbitrage portfolio cannot be nonnegative; that is, we have $S^{\delta} \notin \bbp_{\sfi}^+(\bbs)$. Moreover, in accordance with Theorem \ref{thm-main}, the arbitrage portfolio $S^{\delta}$ requires short selling; that is, we have $S^{\delta} \notin \bbp_{\sfi}^{\delta \geq 0}(\bbs)$. More precisely, as (\ref{strat-Bessel-1}) shows, we have a long position $\theta > 0$ in the risky primary security account $S$. However, by (\ref{strat-Bessel-2}) we have $\eta_0 = -\theta_0 < 0$, and hence, by the continuity of $\eta$, the stopping time $\tau := \inf \{ t \in [0,1] : \eta_t = 0 \}$ is strictly positive, which shows that we have the short position $\eta < 0$ on the stochastic interval $\IL 0,\tau \IL$.
\end{example}

\bibliographystyle{plain}

\bibliography{Finance}

\end{document}